\documentclass[UKenglish]{eptcs}

\usepackage{doc}
\usepackage{makeidx}

\usepackage{hyperref}
\usepackage{amsfonts}
\usepackage{amsmath}
\usepackage{amssymb}
\usepackage{amsthm}
\usepackage{amsbsy}
\usepackage{enumerate}
\usepackage{stackrel}
\usepackage{bm}
\usepackage[toc,page]{appendix}
\usepackage{stmaryrd}
\usepackage{wasysym}
\usepackage{color}
\usepackage{rotating}

\newdimen\proofrulebreadth \proofrulebreadth=.05em
\newdimen\proofdotseparation \proofdotseparation=1.25ex
\newdimen\proofrulebaseline \proofrulebaseline=2ex
\newcount\proofdotnumber \proofdotnumber=3
\let\then\relax
\def\hfi{\hskip0pt plus.0001fil}
\mathchardef\squigto="3A3B
\newif\ifinsideprooftree\insideprooftreefalse
\newif\ifonleftofproofrule\onleftofproofrulefalse
\newif\ifproofdots\proofdotsfalse
\newif\ifdoubleproof\doubleprooffalse
\let\wereinproofbit\relax
\newdimen\shortenproofleft
\newdimen\shortenproofright
\newdimen\proofbelowshift
\newbox\proofabove
\newbox\proofbelow
\newbox\proofrulename
\def\shiftproofbelow{\let\next\relax\afterassignment\setshiftproofbelow\dimen0 }
\def\shiftproofbelowneg{\def\next{\multiply\dimen0 by-1 }%
\afterassignment\setshiftproofbelow\dimen0 }
\def\setshiftproofbelow{\next\proofbelowshift=\dimen0 }
\def\setproofrulebreadth{\proofrulebreadth}

\def\prooftree{
\ifnum  \lastpenalty=1
\then   \unpenalty
\else   \onleftofproofrulefalse
\fi
\ifonleftofproofrule
\else   \ifinsideprooftree
        \then   \hskip.5em plus1fil
        \fi
\fi
\bgroup
\setbox\proofbelow=\hbox{}\setbox\proofrulename=\hbox{}%
\let\justifies\proofover\let\leadsto\proofoverdots\let\Justifies\proofoverdbl
\let\using\proofusing\let\[\prooftree
\ifinsideprooftree\let\]\endprooftree\fi
\proofdotsfalse\doubleprooffalse
\let\thickness\setproofrulebreadth
\let\shiftright\shiftproofbelow \let\shift\shiftproofbelow
\let\shiftleft\shiftproofbelowneg
\let\ifwasinsideprooftree\ifinsideprooftree
\insideprooftreetrue
\setbox\proofabove=\hbox\bgroup$\displaystyle 
\let\wereinproofbit\prooftree
\shortenproofleft=0pt \shortenproofright=0pt \proofbelowshift=0pt
\onleftofproofruletrue\penalty1
}

\def\eproofbit{
\ifx    \wereinproofbit\prooftree
\then   \ifcase \lastpenalty
        \then   \shortenproofright=0pt  
        \or     \unpenalty\hfil         
        \or     \unpenalty\unskip       
        \else   \shortenproofright=0pt  
        \fi
\fi
\global\dimen0=\shortenproofleft
\global\dimen1=\shortenproofright
\global\dimen2=\proofrulebreadth
\global\dimen3=\proofbelowshift
\global\dimen4=\proofdotseparation
\global\count255=\proofdotnumber
$\egroup  
\shortenproofleft=\dimen0
\shortenproofright=\dimen1
\proofrulebreadth=\dimen2
\proofbelowshift=\dimen3
\proofdotseparation=\dimen4
\proofdotnumber=\count255
}

\def\proofover{
\eproofbit 
\setbox\proofbelow=\hbox\bgroup 
\let\wereinproofbit\proofover
$\displaystyle
}%
\def\proofoverdbl{
\eproofbit 
\doubleprooftrue
\setbox\proofbelow=\hbox\bgroup 
\let\wereinproofbit\proofoverdbl
$\displaystyle
}%
\def\proofoverdots{
\eproofbit 
\proofdotstrue
\setbox\proofbelow=\hbox\bgroup 
\let\wereinproofbit\proofoverdots
$\displaystyle
}%
\def\proofusing{
\eproofbit 
\setbox\proofrulename=\hbox\bgroup 
\let\wereinproofbit\proofusing
\kern0.3em$
}

\def\endprooftree{
\eproofbit 
  \dimen5 =0pt
\dimen0=\wd\proofabove \advance\dimen0-\shortenproofleft
\advance\dimen0-\shortenproofright
\dimen1=.5\dimen0 \advance\dimen1-.5\wd\proofbelow
\dimen4=\dimen1
\advance\dimen1\proofbelowshift \advance\dimen4-\proofbelowshift
\ifdim  \dimen1<0pt
\then   \advance\shortenproofleft\dimen1
        \advance\dimen0-\dimen1
        \dimen1=0pt
        \ifdim  \shortenproofleft<0pt
        \then   \setbox\proofabove=\hbox{%
                        \kern-\shortenproofleft\unhbox\proofabove}%
                \shortenproofleft=0pt
        \fi
\fi
\ifdim  \dimen4<0pt
\then   \advance\shortenproofright\dimen4
        \advance\dimen0-\dimen4
        \dimen4=0pt
\fi
\ifdim  \shortenproofright<\wd\proofrulename
\then   \shortenproofright=\wd\proofrulename
\fi
\dimen2=\shortenproofleft \advance\dimen2 by\dimen1
\dimen3=\shortenproofright\advance\dimen3 by\dimen4
\ifproofdots
\then
        \dimen6=\shortenproofleft \advance\dimen6 .5\dimen0
        \setbox1=\vbox to\proofdotseparation{\vss\hbox{$\cdot$}\vss}%
        \setbox0=\hbox{%
                \advance\dimen6-.5\wd1
                \kern\dimen6
                $\vcenter to\proofdotnumber\proofdotseparation
                        {\leaders\box1\vfill}$%
                \unhbox\proofrulename}%
\else   \dimen6=\fontdimen22\the\textfont2 
        \dimen7=\dimen6
        \advance\dimen6by.5\proofrulebreadth
        \advance\dimen7by-.5\proofrulebreadth
        \setbox0=\hbox{%
                \kern\shortenproofleft
                \ifdoubleproof
                \then   \hbox to\dimen0{%
                        $\mathsurround0pt\mathord=\mkern-6mu%
                        \cleaders\hbox{$\mkern-2mu=\mkern-2mu$}\hfill
                        \mkern-6mu\mathord=$}%
                \else   \vrule height\dimen6 depth-\dimen7 width\dimen0
                \fi
                \unhbox\proofrulename}%
        \ht0=\dimen6 \dp0=-\dimen7
\fi
\let\doll\relax
\ifwasinsideprooftree
\then   \let\VBOX\vbox
\else   \ifmmode\else$\let\doll=$\fi
        \let\VBOX\vcenter
\fi
\VBOX   {\baselineskip\proofrulebaseline \lineskip.2ex
        \expandafter\lineskiplimit\ifproofdots0ex\else-0.6ex\fi
        \hbox   spread\dimen5   {\hfi\unhbox\proofabove\hfi}%
        \hbox{\box0}%
        \hbox   {\kern\dimen2 \box\proofbelow}}\doll%
\global\dimen2=\dimen2
\global\dimen3=\dimen3
\egroup 
\ifonleftofproofrule
\then   \shortenproofleft=\dimen2
\fi
\shortenproofright=\dimen3
\onleftofproofrulefalse
\ifinsideprooftree
\then   \hskip.5em plus 1fil \penalty2
\fi
}

\newtheorem{definition}{Definition}[section]
\newtheorem{lemma}[definition]{Lemma}
\newtheorem{proposition}[definition]{Proposition}
\newtheorem{theorem}[definition]{Theorem}

\newtheorem{example}[definition]{Example}

\newtheorem{fact}[definition]{Fact}

\DeclareMathAlphabet{\mathpzc}{OT1}{pzc}{m}{it}

\renewcommand{\vec}[1]{\bm{#1}}
\newcommand{\set}[1]{\{#1\}}

\newcommand{\Comment}[1]{ }

\newcommand{\cons}{\!:\!}

\newcommand{\ored}[1]{\stackrel{#1}{\longrightarrow}}      
 
\newcommand{\rlbk}{\sf rbk}
\newcommand{\rb}{\sf rb}

\newcommand{\CommRule}{\textsf{comm}}

\newcommand{\RbkRule}{\textsf{rbk}}

\newcommand{\comply}{\dashv}
\newcommand{\complyR}{~\begin{sideways}\begin{sideways}$\Vdash$\end{sideways}\end{sideways}~}
\newcommand{\ncomplyR}{\not\!\!\!\complyR}
\newcommand{\complyF}{\comply}

\newcommand{\Bag}{\textsf{bag}}
\newcommand{\DBag}{\Dual{\Bag}}
\newcommand{\Belt}{\textsf{belt}}
\newcommand{\DBelt}{\Dual{\Belt}}
\newcommand{\Price}{\textsf{price}}
\newcommand{\DPrice}{\Dual{\Price}}
\newcommand{\Card}{\textsf{card}}
\newcommand{\DCard}{\Dual{\Card}}
\newcommand{\Cash}{\textsf{cash}}
\newcommand{\DCash}{\Dual{\Cash}}

\newcommand{\Dual}[1]{\overline{#1}}

\newcommand{\der}{\;\vartriangleright\;}
\newcommand{\nder}{\;\not\vartriangleright\;}

\newcommand{\CkptcomplHyp}{\scriptsize \mbox{\sc Hyp}}
\newcommand{\CkptcomplAx}{\mbox{\scriptsize\sc Ax}}

\newcommand{\Set}[1]{\{#1\}}

\newcommand{\Iff}{\text { iff }}
\newcommand{\Impl}{\Rightarrow}

\renewcommand{\implies}{\text{ implies }}

\newcommand{\lts}[1]{\stackrel{#1}{\longrightarrow}}

\newcommand{\rec}{{\sf rec} \, }

\newcommand{\Names}{{\cal N}}
\newcommand{\CoNames}{\overline{\Names}}

\newcommand{\stopA}{{\bf 1}}

\newcommand{\Sbehav}{{\sf RC}}
\newcommand{\SbehavH}{{\sf RCH}}
\newcommand{\Stacks}{{\sf Histories}}

\newcommand{\emptystack}{[\;]}

\newcommand{\back}{\prec}

\newcommand{\np}[2]{#1\back#2}

\newcommand{\pp}{~\|~}

\newcommand{\s}{\Sbehav}

\newcommand{\ctrs}{contracts}

\newcommand{\Ctrs}{Contracts}

\newcommand{\Prove}{\textbf{Prove}}
\newcommand{\IF}{\textbf{if}}
\newcommand{\THEN}{\textbf{then}}
\newcommand{\ELSE}{\textbf{else}}
\newcommand{\AND}{\textbf{and}}

\newcommand{\FAIL}{\textbf{fail}}

\newcommand{\FA}{\textbf{for all}}

\begin{document}

%
\title{Retractable Contracts
\thanks{This work was partially supported by Italian MIUR PRIN
Project CINA Prot.\ 2010LHT4KM and COST Action IC1201 BETTY.}
}


\def\titlerunning{Retractable Contracts
}

%
\author{
Franco Barbanera\institute{Dipartimento di Matematica e Informatica,
  University of Catania,
  \email{barba@dmi.unict.it}} 
\and
Mariangiola Dezani-Ciancaglini\institute{Dipartimento di Informatica,
University of Torino,
\email{dezani@di.unito.it}} \thanks{This author was partially supported by the Torino University/Compagnia San Paolo Project SALT.}
\and
Ivan Lanese\institute{Dipartimento di Informatica - Scienza e Ingegneria,
   University of Bologna/INRIA,
 \email{ivan.lanese@gmail.com}} \thanks{This author was partially supported by the French ANR project REVER n.\ ANR 11 INSE 007 and COST Action IC1405.}
\and
Ugo de'Liguoro\institute{Dipartimento di Informatica,
University of Torino,
\email{deliguoro@di.unito.it}}  \thanks{This author was partially supported by the Torino University/Compagnia San Paolo Project SALT.}
}


\def\authorrunning{Barbanera, Dezani, Lanese and de'Liguoro}

\providecommand{\publicationstatus}{To appear in EPTCS.}

\clearpage

\maketitle

\begin{abstract}
In calculi for modelling communication protocols, internal and external choices play dual roles. Two external choices can be viewed naturally as dual too, as they represent an agreement between the communicating  parties. If the interaction fails, the past agreements are good candidates as points where to roll back, in order to take a different agreement. We propose a variant of contracts with synchronous rollbacks to agreement points in case of deadlock. The new calculus is equipped with a compliance relation which is shown to be decidable. 
\end{abstract}

\thispagestyle{empty}

\section{Introduction}
\label{sect:introduction}

In human as well as  automatic negotiations, an interesting feature is the ability of rolling back to some previous point in case of failure, undoing previous choices and  possibly trying a different path. {\em Rollbacks} are familiar to the users of web browsers, and so are also the troubles that these might cause during ``undisciplined'' interactions. Clicking the ``back'' button, or going to some previous point in the chronology when we are in the middle of a transaction, say the booking of a flight, can be as smart as dangerous. In any case, it is surely a behaviour that service programmers want to discipline. Also the converse has to be treated with care: a server discovering that an auxiliary service becomes available after having started a conversation could take advantage of it using some kind of rollback. However, such a server would be quite unfair if the rollback were completely hidden from the client. 

Let us consider an example.
A $\textsf{Buyer}$ is looking for a bag ($\DBag$) or a belt ($\DBelt)$; she will decide how to pay, either by credit card ($\DCard$) or by cash ($\DCash$), after knowing the $\Price$ from a $\textsf{Seller}$. The $\textsf{Buyer}$ behaviour can be described by the process:
\[\textsf{Buyer} = \DBag.\Price.(\DCard \oplus \DCash) \oplus \DBelt. \Price.(\DCard \oplus \DCash)\]
where dot is sequential composition and $\oplus$ is internal choice. 
The $\textsf{Seller}$ does not accept credit card payments for items of low price, like belts, but only for more expensive ones, like bags:
\[\textsf{Seller} = \Belt. \DPrice.\Cash + \Bag.\DPrice.(\Card + \Cash) \]
where $+$ is external choice.
According to contract theory~\cite{CGP10}, $\textsf{Buyer}$ is not compliant with $\textsf{Seller}$, since she can choose to pay the belt by card. Also, there is no obvious way to represent the buyer's will to be free in her decision about the payment and be compliant with a seller without asking the seller in advance. Nonetheless, when interacting with $\textsf{Seller}$, the buyer's decision is actually free at least in the case of purchase of a bag. For exploiting such a possibility the client (but also the server) should be able to tolerate a partial failure of her protocol, and to try a different path. 

To this aim we add to (some) choices a possibility of rollback, in
case the taken path fails to reach a success configuration. In
this setting, choices among outputs are no more purely internal, since
the environment may oblige to undo a wrong choice and choose a
different alternative. For this reason, we denote choices between outputs which allow rollback as external, hence we use $\textsf{Buyer}'$ below instead of $\textsf{Buyer}$:
\[\textsf{Buyer}' = \DBag.\Price.(\DCard \oplus \DCash) + \DBelt. \Price.(\DCard \oplus \DCash)\]
We thus explore a model of contract interaction in which synchronous
rollback is triggered when client and server fail to reach an agreement.

In defining our model we build over some previous work reported in
\cite{BarbaneraDdL14}, where we have considered contracts with
rollbacks. However, we depart from that model on three main
aspects. First, in the present model rollback is used in a
disciplined way to tolerate failures in the interaction, thus
improving compatibility, while in \cite{BarbaneraDdL14} it is an
internal decision of either client or server, which makes
compatibility more difficult. Second, we embed checkpoints in the
structure of contracts, avoiding explicit checkpoints. Third, we
consider a stack of ``pasts'', called histories, instead of just one
past for each participant, as in~\cite{BarbaneraDdL14}, thus allowing
to undo many past choices looking for a successful alternative.

\section{Contracts for retractable interactions}
\label{sect:contracts}
Our \ctrs\ can be obtained from the session behaviours of~\cite{BdL10} or from the session contracts of~\cite{BH13} just adding external retractable choices between outputs.

\begin{definition}[Retractable \Ctrs]
\label{Adef:ckpt-behav}
Let $\Names$ {\em (set of names)} be some countable set of symbols and $\CoNames = \Set{\Dual{a} \mid a \in \Names}$ {\em (set of conames)}, with
$\Names\cap\CoNames = \emptyset$. 
The set $\s$ of {\bf retractable \ctrs} is defined as the set of the {\bf closed} expressions generated by the following 
grammar, 
\[\begin{array}{lcl@{\hspace{4mm}}l}
\sigma,\rho&:=& \mid ~ \stopA &\mbox{success} \\[1mm]
       &     & \mid ~\sum_{i\in I} a_i.\sigma_i  & \mbox{(retractable) input} \\[1mm]
       &     & \mid ~\sum_{i\in I} \Dual{a}_i.\sigma_i& \mbox{retractable output}\\[1mm]
       &     & \mid ~\bigoplus_{i\in I} \Dual{a}_i.\sigma_i & \mbox{unretractable output}\\[1mm]
       &     & \mid  ~x  & \mbox{variable}\\[1mm]
       &     & \mid ~\rec x. \sigma &  \mbox{recursion}
\end{array}
\]
where $I$ is non-empty and finite, the names 
and  the conames 
in choices are pairwise distinct and $\sigma$ is not a variable in $\rec x.\sigma$.
\end{definition}
\noindent
Note that recursion in $\s$ is guarded and hence contractive in the usual sense. We take an equi-recursive view of recursion by equating $\rec x.\sigma  $ with $\sigma[\rec x.\sigma/x]$. We use $\alpha$ to range over $\Names\cup\CoNames$,
with the convention $\Dual{\alpha}=\begin{cases}
 \Dual{a}     & \text{if }\alpha=a, \\
    a  & \text{if }\alpha=\Dual{a}.
\end{cases}$ \\
We write $\alpha_1.\sigma_1 + \alpha_2.\sigma_2$ for binary
input/retractable output and $\Dual{a}_1.\sigma_1 \oplus
\Dual{a}_2.\sigma_2$ for binary unretractable output. They are both
commutative by definition. Also, $\Dual{a}.\sigma$ may denote both unary
retractable output 
and unary unretractable output. 
This is not a
source of confusion since they have the same semantics.
\medskip

\noindent
From now on we call just {\em \ctrs} the expressions in $\s$. They are written by omitting all trailing $\stopA$'s.

\medskip

In order to deal with rollbacks we decorate \ctrs\ with histories, which memorise the alternatives in choices which have been discharged. We use `$\circ$' as a placeholder for {\em no-remaining-alternatives}.

\begin{definition}[\Ctrs\ with histories]
Let $\Stacks$  be the expressions (referred to also as {\em stacks}) generated by the grammar:
$$\vec{\gamma} ::= \emptystack \mid  \vec{\gamma} \cons  \sigma$$
where $\sigma\in\Sbehav\cup\{\circ\}$ and $\circ\not\in\Sbehav$. Then the set of {\em \ctrs\ with histories} is defined by:
$$\SbehavH = \{ \np{\vec{\gamma}}\sigma \mid \vec{\gamma}\in \Stacks, \sigma\in \Sbehav\cup\{\circ\}\, \}.$$
\end{definition}
\noindent
We write just $\sigma_1\cons \cdots \cons \sigma_k$ for the stack $(\cdots(\emptystack\cons \sigma_1) \cons \cdots )\cons \sigma_{k}$. With a little abuse of notation we use ` $\cons$ ' also to concatenate histories, and to add \ctrs\ in front of histories.

\bigskip

We can now discuss the operational semantics of our calculus (Definition~\ref{scs}). The reduction rule for the internal choice ($\oplus$) is standard, but for the presence of the $\np{\vec{\gamma}}{\cdot}$. Whereas, when reducing retractable choices (+), the discharged branches are memorised. When a single action is executed, the history is modified by adding a `$\circ$', intuitively meaning that the only possible branch has been tried and no alternative is left. Rule $(\rb)$ recovers the contract on the top of the stack, replacing the current one with it.

\begin{definition}[LTS of \Ctrs\ with Histories]\label{scs}
\[\begin{array}{rl@{\hspace{16mm}}ll}
(+) & \np{\vec{\gamma}}{\alpha.\sigma + \sigma'}
			\ored{\alpha} 
	\np{\vec{\gamma} \cons \sigma'}\sigma  & 
(\oplus) &  \np{\vec{\gamma}}{{ \Dual{a}.\sigma \oplus \sigma'}}
			\ored{\tau} 
		\np{\vec{\gamma}}\Dual{a}.\sigma 
\\ [2mm]
(\alpha)
& \np{\vec{\gamma}}\alpha.\sigma \ored{\alpha} \np{\vec{\gamma}\cons\circ}{\sigma}
&
(\rb)  &
	\np{\vec{\gamma}\cons\sigma'}{\sigma}  \lts{\rb}  \np{\vec{\gamma}}\sigma'
\end{array}\]
\end{definition}

The interaction of a client with a server is modelled by the reduction of their parallel composition, that can be either
forward, consisting of CCS style synchronisations and single internal choices, or backward, only when there is no possible forward reduction, and the client is not satisfied, i.e. it is different from $\stopA$. 

\begin{definition}[LTS of Client/Server Pairs]\label{def:cspairlts} 
We define the relation $\ored{}$ over pairs of  \ctrs\ with histories by the following rules:
\[\begin{array}{c}
\prooftree
\np{\vec{\delta}}\rho \ored{\alpha} \np{\vec{\delta'}}{\rho'} \qquad \np{\vec{\gamma}}\sigma \ored{\Dual{\alpha}} \np{\vec{\gamma'}}{\sigma'} 
\justifies
\np{\vec{\delta}}\rho\pp \np{\vec{\gamma}}\sigma \ored{} \np{\vec{\delta'}}{\rho'}\pp \np{\vec{\gamma'}}{\sigma'}
\using (\text{\em \CommRule})
\endprooftree \\[8mm] 
\prooftree
\np{\vec{\delta}}\rho \ored{\tau} \np{\vec{\delta} }{\rho'} 
\justifies
\np{\vec{\delta}}\rho\pp \np{\vec{\gamma}}\sigma \ored{} \np{\vec{\delta}}{\rho'}\pp \np{\vec{\gamma}}{\sigma}
\using (\tau)
\endprooftree\\[8mm] 
\prooftree
\np{\vec{\gamma}}{\rho} \lts{\rb} \np{\vec{\gamma'}}{\rho'} \qquad \np{\vec{\delta}}{\sigma} \lts{\rb} \np{\vec{\delta'}}{\sigma'} 
	\qquad  \rho\neq\stopA
\justifies
\np{\vec{\gamma}}{\rho}\pp\np{\vec{\delta}}{\sigma} \ored{} \np{\vec{\gamma'}}{\rho'}\pp\np{\vec{\delta'}}{\sigma'}
\using (\text{\em \RbkRule})
\endprooftree 
\end{array}\]
plus the rule symmetric to $(\tau)$ w.r.t. $\|$. Moreover, rule $(\text{\em \RbkRule})$ applies only if neither $(\text{\em \CommRule})$ nor $(\tau)$ do.
\end{definition}
\noindent
We will use $\ored{*}$ and $\not\!\!\ored{}$ with the standard meanings. \\
Notice that, since `$\circ$' cannot synchronise with anything, in case a partner rolls back to a `$\circ$', it is forced to recover an {\em older} past (if any).   

\medskip
The following examples show the different behaviours of retractable and unretractable outputs. We decorate arrows with the name of the used reduction rule. 
As a first example we consider a possible reduction of the process discussed in the Introduction.
\begin{example}\label{example1}
{\em As in the Introduction, let
$\textsf{Buyer}' = \DBag.\Price.(\DCard \oplus \DCash) + \DBelt. \Price.(\DCard \oplus \DCash)$ be a client and $\textsf{Seller} = \Belt. \DPrice.\Cash + \Bag.\DPrice.(\Card + \Cash)$ a server; then
\[\begin{array}{lrcl}
  &
\np{\emptystack}{\textsf{Buyer}'}
 & \pp & \np{\emptystack}{\textsf{Seller}}\\[2mm]

 \ored{\CommRule} &
	\np{\DBag.\Price.(\DCard \oplus \DCash)}{\Price.(\DCard \oplus \DCash)}
& \pp & \np{\Bag.\DPrice.(\Card + \Cash)}{\DPrice.\Cash}\\[2mm]

 \ored{\CommRule} &
	\np{\DBag.\Price.(\DCard \oplus \DCash)\cons \circ }{(\DCard \oplus \DCash)}
 & \pp & \np{\Bag.\DPrice.(\Card + \Cash)\cons \circ }{\Cash}\\[2mm]

 \ored{\tau} &
	\np{\DBag.\Price.(\DCard \oplus \DCash)\cons \circ }{\DCard}
& \pp & \np{\Bag.\DPrice.(\Card + \Cash)\cons \circ }{\Cash}\\[2mm]

\ored{\RbkRule} &
	\np{\DBag.\Price.(\DCard \oplus \DCash)}{\circ }
 & \pp & \np{\Bag.\DPrice.(\Card + \Cash)}{\circ }\\[2mm]

\ored{\RbkRule} &
	\np{\emptystack}{\DBag.\Price.(\DCard \oplus \DCash) }
 & \pp & \np{\emptystack}{\Bag.\DPrice.(\Card + \Cash)}\\[2mm]

\ored{\CommRule} &
	\np{\circ}{\Price.(\DCard \oplus \DCash) }
 & \pp & \np{\circ}{\DPrice.(\Card + \Cash)}\\[2mm]

\ored{\CommRule} &
	\np{\circ\cons\circ}{(\DCard \oplus \DCash) }
 & \pp & \np{\circ\cons\circ}{(\Card + \Cash)}\\[2mm]

\ored{\tau} &
	\np{\circ\cons\circ}{\DCard }
 & \pp & \np{\circ\cons\circ}{(\Card + \Cash)}\\[2mm]

\ored{\CommRule} &
	\np{\circ\cons\circ\cons\circ}{\stopA }
 & \pp & \np{\circ\cons\circ\cons\Cash}{\stopA}\\[2mm]

\,\,\not\!\!\ored{}
	
\end{array}\]
}
\end{example}

\begin{example}\label{exampleRed1}
{\em
Let $\rho = \rec x.(\Dual{b}.x \oplus \Dual{a}.c.x)$ and  
     $\sigma = {\rec x.(b.x + a.\Dual{e}.x)}$. The following reduction sequence leads the parallel composition of these contracts to a deadlock.
 \[\begin{array}{llrcl}
\rho\pp\sigma&
 \ored{\tau} &
	\np{\emptystack}{\Dual{a}.c.\rho}
& \pp &\np{\emptystack}{\rec x.(b.x + a.\Dual{e}.x)}\\[2mm]
&\ored{\CommRule} &
	\np{\circ}{c.\rho}
& \pp & \np{b.\sigma }{\Dual{e}.\sigma}\\[2mm]
&
 \ored{\RbkRule } &
	\np{\emptystack}{\circ}
& \pp & \np{\emptystack }{b.\sigma }\\[2mm]
&
\,\,\not\!\!\ored{ } 
\end{array}\]
}
\end{example}

\begin{example}\label{exampleRed2}
{\em
Let us now modify the above example by using retractable outputs in the client, so making the two contracts in parallel always reducible. The following reduction shows that there can be an infinite number of  rollbacks in a sequence, even if it is not possible to have an infinite reduction
containing only rollbacks.
Notice how the stack keeps growing indefinitely.

\noindent
Let $\rho = \rec x.(\Dual{b}.x + \Dual{a}.c.x)$ and  
     $\sigma = {\rec x.(b.x + a.\Dual{e}.x)}$.
\[\begin{array}{llrcl}
  \rho\pp\sigma& \ored{\CommRule} &
	\np{\Dual{b}.\rho}{c.\rho}
& \pp &\np{b.\sigma}{\Dual{e}.\sigma}\\[2mm]
&
 \ored{ \RbkRule} &
	\np{\emptystack}{\Dual{b}.\rho}
& \pp & \np{\emptystack }{b.\sigma}\\[2mm]
&
 \ored{\CommRule } &
	\np{\circ}{\rho}
& \pp & \np{ \circ}{\sigma}\\[2mm]
&
 \ored{\CommRule} &
	\np{\circ\cons \Dual{b}.\rho}{c.\rho}
& \pp &\np{\circ\cons b.\sigma}{\Dual{e}.\sigma}\\[2mm]
&
 \ored{ \RbkRule} &
	\np{\circ}{\Dual{b}.\rho}
& \pp & \np{\circ }{b.\sigma}\\[2mm]
&
 \ored{\CommRule } &
	\np{\circ\cons\circ}{\rho}
& \pp & \np{ \circ\cons\circ}{\sigma}\\[2mm]
&
  & 
	 & \vdots
\end{array}\]
}
\end{example}

\section{Compliance}

The compliance relation for standard \ctrs\  consists in requiring that, whenever no reduction is possible, all client requests and offers have been satisfied, i.e.\ the client is in the success state $\stopA$. For  retractable \ctrs\ we can adopt the same definition. 

\begin{definition}[Compliance Relation $\complyR$]\label{ccr}\hfill
\begin{enumerate}[i)]
\item \label{c1}
The relation $\complyR$ on \ctrs\ with histories is defined by:\\
\centerline{$\begin{array}{ll}\np{\vec{\delta}}\rho\complyR \np{\vec{\gamma}}\sigma  \text{ whenever }
\np{\vec{\delta}}\rho\pp \np{\vec{\gamma}}\sigma
                    \ored{*}
                     \np{\vec{\delta'}}\rho'\pp \np{\vec{\gamma'}}\sigma'
                     \not\!\!\ored{} 
\mbox{ implies }\rho'=\stopA\\ \hfill\text{ for any } 
\vec{\delta'},\rho',\vec{\gamma'},\sigma'.
\end{array}$}
\item\label{c2} The relation $\complyR$ on \ctrs\ is defined by:\\
\centerline{$\rho\complyR\sigma ~~~~\text{if}~~~~ \np\emptystack\rho\complyR \np\emptystack\sigma.
$}
\end{enumerate}

\end{definition}

We now provide a formal system characterising 
compliance on retractable \ctrs.\\
The judgments are of the shape $\Gamma\der\rho \complyF \sigma$, 
where $\Gamma$ is a set of expressions of the form
$\rho' \complyF \sigma'$.
We write $\der\rho \complyF \sigma$ when $\Gamma$ is empty. The only non standard rule is rule $(+,+)$, which assures compliance of two retractable choices when they contain respectively a name and the corresponding coname followed by compliant contracts. This contrasts with the rules $(\oplus,+)$ and $(+,\oplus)$, where all conames in unretractable choices between outputs must have corresponding names in the choices between inputs, followed by compliant contracts. 

\begin{definition}[Formal System for Compliance]\label{def:formalCompl}
\[\begin{array}{c@{\hspace{8mm}}c@{\hspace{8mm}}c}
\prooftree 
\justifies
	\Gamma\der \stopA \complyF \sigma
\using(\CkptcomplAx)
\endprooftree 
&
\prooftree
\justifies
\Gamma, \rho\complyF\sigma \der \rho\complyF\sigma 
\using(\CkptcomplHyp)
\endprooftree 
&
\prooftree  
     \Gamma, \alpha.\rho+\rho'\complyF\Dual{\alpha}.\sigma+\sigma'
     	\der \rho
    	  \complyF
		\sigma
\justifies
    \Gamma\der \alpha.\rho+\rho'\complyF\Dual{\alpha}.\sigma+\sigma'
\using(+,+)
\endprooftree 
\end{array}\]
\vspace{2mm}
\[\begin{array}{c}
\prooftree 
    \forall i\in I.~ \Gamma,\mbox{\small $\bigoplus$}_{i\in I} \Dual{a}_i.{\rho}_i\complyF
    	\mbox{\small $\sum$}_{j\in I\cup J} a_j.{\sigma}_j\der 
    	\rho_i
    	\complyF
    		\sigma_i
  \justifies
    \Gamma\der \mbox{\small $\bigoplus$}_{i\in I} \Dual{a}_i.{\rho}_i\complyF
    	\mbox{\small $\sum$}_{j\in I\cup J} a_j.{\sigma}_j
\using (\oplus, +)
\endprooftree 
\\[10mm]
\prooftree 
    \forall i\in I.~ \Gamma,\mbox{\small $\sum$}_{j\in I\cup J} a_j.{\sigma}_j\complyF \mbox{\small $\bigoplus$}_{i\in I} \Dual{a}_i.{\rho}_i \der 
    	\rho_i
    	\complyF
    		\sigma_i
  \justifies
    \Gamma\der\mbox{\small $\sum$}_{j\in I\cup J} a_j.{\sigma}_j\complyF \mbox{\small $\bigoplus$}_{i\in I} \Dual{a}_i.{\rho}_i    	
\using (+,\oplus)
\endprooftree 
\end{array}\]
\end{definition}
\noindent
Notice that rule $(+,+)$ implicitly represents the fact that, in the decision procedure for two contracts made of retractable choices, 
the possible synchronising branches have to be tried, until either a successful one is found or all fail.

\begin{example}
{\em
Let us formally show that, for the 
$\textsf{Buyer}'$ and $\textsf{Seller}$ of the Introduction, we have $\textsf{Buyer}'\complyF\textsf{Seller}$.\\
For the sake of readability, let\\
$\Gamma' = \textsf{Buyer}'  \complyF \textsf{Seller},\ \Price.(\DCard \oplus \DCash) \complyF \DPrice.(\Card + \Cash)$ and
 $\Gamma''= \Gamma' ,\ \DCard \oplus \DCash \complyF \Card + \Cash$

{\small
$$
\prooftree
 \prooftree
      \prooftree
             \prooftree 
             \justifies
                    \Gamma'' \der \stopA \complyF \stopA
             \using(\CkptcomplAx)
             \endprooftree
\qquad
             \prooftree 
             \justifies
                    \Gamma'' \der \stopA \complyF \stopA
             \using(\CkptcomplAx)
             \endprooftree
      \justifies
            \Gamma'  \der \DCard \oplus \DCash \complyF \Card + \Cash
       \using(\oplus,+)
       \endprooftree
   \justifies
       \textsf{Buyer}'  \complyF \textsf{Seller} \der \Price.(\DCard \oplus \DCash) \complyF \DPrice.(\Card + \Cash)
    \using(+,+)
   \endprooftree
\justifies
     \der \textsf{Buyer}'  \complyF \textsf{Seller} 
\using(+,+)
\endprooftree
$$
}
}
\end{example}

\begin{example}
{\em
The contracts of Example~\ref{exampleRed2} can be formally proved to be compliant by means
of the following derivation in our formal system. Actually such a derivation can be looked at as
the result of the decision procedure implicitly described by the formal system.
$$
\prooftree
 \prooftree
   \justifies
      \Dual{b}.\rho + \Dual{a}.c.\rho \complyF  b.\sigma + a.\Dual{e}.\sigma \der \rho \complyF \sigma
   \using(\CkptcomplHyp)
   \endprooftree
\justifies
     \der \Dual{b}.\rho + \Dual{a}.c.\rho \complyF  b.\sigma + a.\Dual{e}.\sigma
\using(+,+)
\endprooftree
$$
In applying the rules we exploit the fact that we consider contracts modulo recursion fold/unfold. 
}\end{example}

We can show that derivability in 
this formal system is decidable, 
since it is syntax 
directed and it does not admit infinite derivations.\\

We denote by ${\cal D}$ a derivation in the system of Definition~\ref{def:formalCompl}.
The procedure \Prove\; in Figure~\ref{fig:Prove} clearly implements the formal system,
that is it is straightforward to check the following
\begin{fact}\label{fact:provecorr}\hfill
\begin{enumerate}[i)]
\item\vspace{-2mm}~~~~
{\em \Prove}$(\Gamma\der \rho\complyF\sigma)\neq$ {\em \FAIL}~~~~\Iff~~~~$\Gamma\der \rho\complyF\sigma$.
\item\vspace{-3mm}~~~~
{\em \Prove}$(\Gamma\der \rho\complyF\sigma)={\cal D}\neq$ {\em \FAIL}$~~~~\implies~~~~\begin{array}{c}{\cal D}\\ \Gamma\der \rho\complyF\sigma\end{array}$
\end{enumerate}
\end{fact}

\begin{theorem}\label{the:a}
Derivability in the formal system is decidable.
\end{theorem}
\begin{proof}
 By Fact \ref{fact:provecorr}, we only need to show that the procedure \Prove\; always terminates. Notice that
in all recursive calls \Prove$(\Gamma, \rho\complyF\sigma \!\der\! \rho_k\complyF\sigma_k)$
inside
\Prove$(\Gamma\der \rho\complyF\sigma)$ the expressions $\rho_k$ and  $\sigma_k$ are subexpressions of $\rho$ and $\sigma$
respectively
(because of unfolding of recursion they can also be $\rho$ and $\sigma$). 
Since contract expressions generate regular trees, there are only finitely many such subexpressions. This implies that the number of different calls of procedure \Prove\ is always finite. 
\end{proof}
\begin{figure}[ht]
  
  \centering
{\small
\begin{tabbing}
\Prove\=$(\Gamma\der \rho\complyF\sigma)$ \\ [1mm]
\IF \> $\rho = \stopA$~~ \THEN ~~~~~
	$\prooftree 
	\justifies
		\Gamma\der \stopA \complyF \sigma
	\using(\CkptcomplAx)
	\endprooftree $\\ [1mm]
\ELSE \> \IF \= ~~~$\rho\complyF\sigma \in \Gamma$ ~~\THEN~~~~~
	$\prooftree
	\justifies
		\Gamma, \rho\complyF\sigma \der \rho\complyF\sigma 
	\using(\CkptcomplHyp)
	\endprooftree $ \\ [1mm]
\ELSE \> \IF ~~\=~~ \= $\rho = \sum_{i\in I}\alpha_i.\rho_i$ ~\= \AND~$\sigma = \sum_{j\in J}\Dual{\alpha}_j.\sigma_j$ \\[2mm]
\> \> \> \AND~~~{\bf exists} $k \in I\cap J$ {\bf s.t.} ${\cal D} = $ \Prove$(\Gamma, \rho\complyF\sigma \der \rho_k\complyF\sigma_k) \neq \FAIL$ \\ [2mm]
\>  \THEN  ~~~~~
	$\prooftree  
     		{\cal D}
	\justifies
    		\Gamma\der \rho\complyF\sigma
	\using(+,+)
	\endprooftree$ ~~~~\ELSE ~~\FAIL\\ [2mm]
\ELSE \> \IF \> ~~~$\rho = \bigoplus_{i\in I}\Dual{a}_i.\rho_i$ ~\AND~ $\sigma = \sum_{j\in J}a_j.\sigma_j$~\AND~ $I\subseteq J$\\[2mm]
\>\>\> \AND~ \FA~$k\in I$ ~~${\cal D}_k = $~\Prove$(\Gamma, \rho\complyF\sigma \der \rho_k\complyF\sigma_k) \neq \FAIL$\\ [2mm]
\>  \THEN ~~~~~
	$\prooftree 
    		\forall k\in I ~~{\cal D}_k 
 	 \justifies
    		\Gamma\der \rho\complyF\sigma
	\using (\oplus, +)
	\endprooftree $ \\ [1mm]
\ELSE \> \IF \> ~~~$\rho = \sum_{i\in I}a_i.\rho_i$ ~\AND~ $\sigma = \bigoplus_{j\in J}\Dual{a}_j.\sigma_j$~\AND~ $I\supseteq J$\\ [2mm]
\>\>\> \AND~~\FA~$k\in J$~~~${\cal D}_k = $~\Prove$(\Gamma, \rho\complyF\sigma \der \rho_k\complyF\sigma_k)\neq \FAIL$\\ [2mm]
\>  \THEN ~~~~~
	$\prooftree 
    		\forall k\in J ~~{\cal D}_k 
 	 \justifies
    		\Gamma\der \rho\complyF\sigma
	\using (+,\oplus)
	\endprooftree $  ~~~~\ELSE ~~\FAIL\\ [1mm]
\ELSE \> \FAIL
\end{tabbing}
\vspace{-4mm}
}\caption{The procedure\;\Prove.}\label{fig:Prove}
\end{figure}

In the remaining of this section we will show the soundness and the completeness of the formal system using some auxiliary lemmas.  

\paragraph{Soundness}
\label{sect:easychair-requirements}
It is useful to show that if a configuration is stuck, then both histories are empty. This is a consequence of the fact that the property ``the histories of client and server have the same length'' is preserved by reductions. 

\begin{lemma}\label{lem:stack-len}If
	$\np{\vec{\delta}}{\rho'} \pp \np{\vec{\gamma}}{\sigma'} \not\!\!\ored{}$, then	 $\vec{\delta} = \vec{\gamma} = \emptystack$.
\end{lemma}

\begin{proof} 
Clearly $\np{\vec{\delta}}{\rho'} \pp \np{\vec{\gamma}}{\sigma'} \not\!\!\ored{}$ implies either $\vec{\delta} =\emptystack$ or $\vec{\gamma} = \emptystack$.
Observe that:
\begin{itemize}
\item rule  $(\CommRule)$ adds one element to both stacks;
\item rule  $(\tau)$ does not modify both stacks;
\item rule  $(\RbkRule)$ removes one element from both stacks.
\end{itemize}
Then starting from two stacks containing the same number of elements, the reduction always produces two stacks containing the same number of elements. So $\vec{\delta} =\emptystack$ implies $\vec{\gamma} = \emptystack$ and vice versa.
\end{proof}

Next we state a lemma relating the logical rules of the formal system for compliance to the reduction rules; notice that the formers do not mention stacks in their judgments.

\begin{lemma}\label{lem:computation}
If the following is an instance of rule $(+,+)$, or $(\oplus,+)$, or $(+,\oplus)$:
\[
\prooftree
	\Gamma, \rho \complyF \sigma \der \rho_1 \complyF \sigma_1 \quad \cdots \quad
	\Gamma, \rho \complyF \sigma \der \rho_n \complyF \sigma_n
\justifies
	\Gamma \der \rho \complyF \sigma
\endprooftree
\]
then for all $\vec{\delta},\vec{\gamma}$ and for all $i = 1, \ldots, n$ there exist $\vec{\delta}_i,\vec{\gamma}_i$ such that
\[ \np{\vec{\delta}}{\rho} \pp \np{\vec{\gamma}}{\sigma} ~ \ored{*} ~ \np{\vec{\delta}_i}{\rho_i} \pp \np{\vec{\gamma}_i}{\sigma_i}\]
and 
rule {\em $(\RbkRule)$} is not used, namely no rollback occurs. 
\end{lemma}

\begin{proof} By inspection of the deduction and reduction rules. 
\end{proof}

\begin{theorem}[Soundness]
If $\der\rho\complyF \sigma$, then $\rho\complyR\sigma$.
\end{theorem}
\begin{proof}
The proof is by contradiction. Assume $\rho\ncomplyR\sigma$. Then
there is a reduction \[\np{\emptystack}{\rho} \pp
\np{\emptystack}{\sigma} \ored{*} \np{\emptystack}{\rho'} \pp
\np{\emptystack}{\sigma'} \not\!\!\ored{}\] with $\rho' \neq \stopA$. Note
that both the histories are empty thanks to
Lemma~\ref{lem:stack-len}. We proceed by induction on the number $n$
of steps in the reduction. 

Let us consider the base case ($n=0$). In this case $\rho \neq \stopA$
and there is no possible synchronization. Rule $\CkptcomplAx$ is not
applicable since $\rho \neq \stopA$. Rule $\CkptcomplHyp$ is not
applicable since $\Gamma$ is empty. The other rules are not applicable otherwise we would have a possible synchronization.

Let us consider the inductive case. We have a case analysis on the
topmost operators in $\rho$ and $\sigma$. Let us start with the case
where both topmost operators are retractable sums, i.e., $\rho =
a.\rho_k + \rho''$ and $\sigma = \Dual{a}.\sigma_k + \sigma''$. Thus,
\[\np{\emptystack}{a.\rho_k + \rho''} \pp
\np{\emptystack}{\Dual{a}.\sigma_k + \sigma''} \ored{*}
\np{\emptystack}{\rho''} \pp \np{\emptystack}{\sigma''} \ored{*}
\np{\emptystack}{\rho'} \pp \np{\emptystack}{\sigma'} \not\ored{}.\] By
definition $\rho'' \ncomplyR \sigma''$, and since this requires a reduction of length $ < n$, by inductive hypothesis
$\nder\rho'' \complyF \sigma''$. Also the above reduction begins by:
\[\np{\emptystack}{a.\rho_k + \rho''} \pp
\np{\emptystack}{\Dual{a}.\sigma_k + \sigma''} \ored{}
\np{\rho''}{\rho_k} \pp \np{\sigma''}{\sigma_k}\ored{*}
\np{\rho''}{\rho'_k} \pp \np{\sigma''}{\sigma'_k} \ored{}
\np{\emptystack}{\rho''} \pp \np{\emptystack}{\sigma''} 
\] 
 for some $\rho'_k, \sigma'_k$. This implies, by the conditions on rule $(\RbkRule)$ and by Lemma \ref{lem:stack-len}:
\[\np{\emptystack}{\rho_k} \pp
\np{\emptystack}{\sigma_k} \ored{*} \np{\emptystack}{\rho'_k} \pp
\np{\emptystack}{\sigma'_k} \not\ored{}.\] 
It follows that $\rho_k
\ncomplyR \sigma_k$, and by inductive hypothesis $\nder\rho_k \complyF
\sigma_k$. We now claim that:
\[
\nder\rho'' \complyF \sigma'' ~\text{and}~\nder\rho_k \complyF 
\sigma_k ~~ \text{imply} ~~  \nder a.\rho_k + \rho''  \complyF \Dual{a}.\sigma_k + \sigma''
\]
Toward a contradiction let us assume that $\der a.\rho_k + \rho''  \complyF \Dual{a}.\sigma_k + \sigma''$; then the only applicable rule is
$(+,+)$, which requires
both 
\[a.\rho_k + \rho''  \complyF \Dual{a}.\sigma_k + \sigma'' \der  \rho''  \complyF \sigma''
~~\mbox{and}~~
a.\rho_k + \rho''  \complyF \Dual{a}.\sigma_k + \sigma'' \der \rho_k  \complyF \sigma_k.\]
Because the only difference between these statements and $\der\rho'' \complyF \sigma''$
and $\der\rho_k \complyF \sigma_k$ is the assumption 
$a.\rho_k + \rho''  \complyF \Dual{a}.\sigma_k + \sigma''$, which is used only in rule $\CkptcomplHyp$,
there must be at least one branch of the derivation tree of $\der a.\rho_k + \rho''  \complyF \Dual{a}.\sigma_k + \sigma''$
ending by such a rule. By Lemma \ref{lem:computation} this implies that 
$\np{\emptystack}{a.\rho_k + \rho''} \pp
\np{\emptystack}{\Dual{a}.\sigma_k + \sigma''} \ored{*} \np{\emptystack}{a.\rho_k + \rho''} \pp
\np{\emptystack}{\Dual{a}.\sigma_k + \sigma''}$ by a reduction never using rule $(\RbkRule)$. By definition
this implies that $\np{\emptystack}{a.\rho_k + \rho''} \complyR \np{\emptystack}{\Dual{a}.\sigma_k + \sigma''}$, contradicting the
hypothesis.

All other cases are similar.
%
%
\end{proof}

\Comment{
In the following we deal with {\em reduction trees}, i.e. with trees of reduction steps on client/server pairs.
\begin{definition}
A reduction tree is {\em $(+.+)$-complete} if any node with more than one child is of the form
 $\mbox{\small $\sum$}_{i\in I} \Dual{\alpha}_i.{\rho}_i \pp
\mbox{\small $\sum$}_{j\in J} \alpha_j.{\sigma}_j$ and such that the set of
its children is $\{ {\rho}_k \!\pp\! \alpha_k \mid k\in (I\cup J) \}.$
\end{definition}
\begin{lemma}
\label{lem:sslemma}
Let $$\np{\emptystack}\rho\pp \np{\emptystack}\sigma
                    \ored{*}
                     \np{\vec{\delta}}\rho'\pp \np{\vec{\gamma}}\sigma'
                     \not\!\!\ored{} $$
with $\rho'\neq\stopA$
Then there exists a $(+.+)$-complete reduction tree with root $\np{\emptystack}\rho\pp \np{\emptystack}\sigma$ containing no $\ored{\RbkRule}$ reduction step and such that for any leaf $\np{\vec{\delta'}}\rho''\pp \np{\vec{\gamma''}}\sigma''$ we have
$\rho''\neq\stopA$ and $\np{\vec{\delta'}}\rho''\pp \np{\vec{\gamma''}}\sigma''\ \ \not\!\!\!\!\ored{\CommRule,\tau}$.
\end{lemma}
\Comment{
\begin{proof}
By induction on the number $n$ of $\ored{\RbkRule}$ reduction steps in
$\np{\emptystack}\rho\pp \np{\emptystack}\sigma
                    \ored{*}
                     \np{\vec{\delta}}\rho'\pp \np{\vec{\gamma}}\sigma'
                     \not\!\!\ored{} $.\\
Let then $n=1$ with\\
$\np{\emptystack}\rho\pp \np{\emptystack}\sigma
                    \ored{*}
                     \np{\vec{\delta_1}\cons\rho'_1}\rho_1\pp \np{\vec{\gamma_1}\cons\sigma'_1}\sigma_1
                     \ored{\RbkRule}
                     \np{\vec{\delta_1}}\rho'_1\pp \np{\vec{\gamma_1}}\sigma'_1
                     \ored{*}
                     \np{\vec{\delta}}\rho'\pp \np{\vec{\gamma}}\sigma'
                     \not\!\!\!\ored{} $.\\
Since we start from an empty stack, let us then consider the $\ored{\CommRule}$ reduction where $\rho'_1$ and $\sigma'_1$ have been pushed on the stack.\\
$\np{\vec{\delta_1}}\rho_2\pp \np{\vec{\gamma_1}}\sigma_2
                    \ored{\CommRule}
                     \np{\vec{\delta_1}\cons\rho'_1}\rho'_2\pp \np{\vec{\gamma_1}\cons\sigma'_1}\sigma'_2$.\\
We have now to take into account several cases (we shall ignore the symmetric ones):
\begin{description}
\item $\rho_2 = \Dual{a}.\rho'_2$, $\sigma_2 = a.\sigma'_2$

\end{description}

\end{proof}
}

\begin{proposition}[Soundness]
For any judgment $\der \rho\complyF\sigma$:
\[\der\rho\complyF \sigma~~ \Impl ~~ \rho\complyR\sigma.\]
\end{proposition}
\begin{proof}
By contraposition, let $\not\models \rho\complyF\sigma$.
This means that there exists a reduction sequence such that
$$\np{\emptystack}\rho\pp \np{\emptystack}\sigma
                    \ored{*}
                     \np{\vec{\delta'}}\rho'\pp \np{\vec{\gamma'}}\sigma'
                     \not\!\!\ored{} $$
with $\rho\neq\stopA$.
Let $T$ be the reduction tree obtained by Lemma \ref{lem:sslemma}
out of such a sequence. From $T$ it is possible to build a subtree of the tree of the recursive calls of the
{\bf Prove} algorithm. Moreover such a subtree is a \FAIL\, subtree, resulting hence in a failure
for {\bf Prove}($\der \rho\complyF\sigma$).This means that $\not\!\der\rho\complyF\sigma$.
\end{proof}
}

\paragraph{Completeness}
The following lemma proves that compliance is preserved by the concatenation of histories to the left of the current histories.

\begin{lemma}\label{lem:hypSoundness} If
$\np {\vec{\delta}}{\rho} \complyR \np {\vec{\gamma}}{\sigma}$, then 
	 $\np {\vec{\delta'} \cons \vec{\delta}}{\rho} \complyR \np {\vec{\gamma'} \cons \vec{\gamma}}{\sigma}$
	 for all $\vec{\delta'}$ , $\vec{\gamma'}$. 
\end{lemma}

\begin{proof}
It suffices to show that
\[\np {\vec{\delta}}{\rho} \complyR \np {\vec{\gamma}}{\sigma} \implies 
	\np {\rho' \cons \vec{\delta}}{\rho} \complyR \np {\vec{\gamma}}{\sigma} \text{ and }
	\np { \vec{\delta}}{\rho} \complyR \np {\sigma' \cons \vec{\gamma}}{\sigma}
\]
which we prove by contraposition.

Suppose that $\np {\rho' \cons \vec{\delta}}{\rho} \ncomplyR \np {\vec{\gamma}}{\sigma}$; then
\[\np {\rho' \cons \vec{\delta}}{\rho} \pp \np {\vec{\gamma}}{\sigma} \ored{*} \np {\vec{\delta}'}{\rho''} \pp \np {\vec{\gamma}'}{\sigma''}
\not\!\!\ored{} \text{ and } \rho''\neq \stopA
\]
If $\rho'$ is never used, then  $\vec{\delta}' = \rho' \cons \vec{\delta}''$ and $\vec{\gamma}'=\emptystack$, so that
we get\[\np {\vec{\delta}}{\rho} \pp \np {\vec{\gamma}}{\sigma} 
\ored{*} \np {\vec{\delta}''}{\rho''} \pp \np \emptystack{\sigma''} \not\!\!\ored{}\]
Otherwise we have that
\[\np {\rho' \cons \vec{\delta}}{\rho} \pp \np {\vec{\gamma}}{\sigma} \ored{*} \np {\rho'}{\rho''} \pp \np {\vec{\gamma}'}{\sigma''}
\ored{} \np\emptystack\rho'\pp\np {\vec{\gamma}''}{\sigma'''}\]
and we assume that $\ored{*}$ is the shortest such reduction.
It follows that $\rho''\neq \stopA $. By the minimality assumption about the length of $\ored{*}$ we know
that $\rho'$ neither has been restored by some previous application of rule $(\rlbk)$, nor pushed back into the stack before. We get
\[\np {\vec{\delta}}{\rho} \pp \np {\vec{\gamma}}{\sigma} \ored{*} \np {\emptystack}{\rho''} \pp \np {\vec{\gamma}''}{\sigma''} \not\!\!\ored{}\]
In both cases we conclude that $\np {\vec{\delta}}{\rho} \ncomplyR \np {\vec{\gamma}}{\sigma}$ as desired. 

\medskip

Similarly we can show that $\np { \vec{\delta}}{\rho} \ncomplyR \np {\sigma' \cons \vec{\gamma}}{\sigma} \implies
\np {\vec{\delta}}{\rho} \ncomplyR \np {\vec{\gamma}}{\sigma}$.
\end{proof}

The following lemma gives all possible shapes of compliant contracts. It is the key lemma for the proof of completeness.
\begin{lemma}\label{lem:coinductiveChar}
If ${\rho} \complyR {\sigma}$, then one of the following conditions holds:
\begin{enumerate}
\item \label{lem:coinductiveChar-1} $\rho = \stopA$;
\item \label{lem:coinductiveChar-2} $\rho = \sum_{i\in I}\alpha_i.\rho_i$, $\sigma = \sum_{j\in J}\Dual{\alpha}_j.\sigma_j$ and 
	$\exists k \in I \cap J.\; {\rho_k} \complyR {\sigma_k}$;
\item \label{lem:coinductiveChar-3} $\rho = \bigoplus_{i\in I}\Dual{a}_i.\rho_i$, $\sigma = \sum_{j\in J}a_j.\sigma_j$,
	$I\subseteq J$ and $\forall k \in I. \; {\rho_k} \complyR {\sigma_k}$;
\item \label{lem:coinductiveChar-4} $\rho = \sum_{i\in I}a_i.\rho_i$, $\sigma = \bigoplus_{j\in J}\Dual{a}_j.\sigma_j$,
	$I\supseteq J$ and $\forall k \in J. \; {\rho_k} \complyR {\sigma_k}$.
\end{enumerate}
\end{lemma}

\begin{proof}
By contraposition and by cases of the possible shapes of $\rho$ and $\sigma$. 

\bigskip

Suppose $\rho = \sum_{i\in I}\alpha_i.\rho_i$, $\sigma = \sum_{j\in J}\Dual{\alpha}_j.\sigma_j$, $I \cap J=\set{k_1,\ldots,k_n}$ and 
$ \rho_{k_i} \ncomplyR \sigma_{k_i}$ for $1\leq i\leq n$. Then we get 
\[\np\emptystack\rho_{k_i} \pp \np\emptystack\sigma_{k_i} \ored{*} \np{\vec{\delta}_i}{\rho'_i} \pp \np{\vec{\gamma}_i}{\sigma'_i} \not\!\!\ored{}\]
for $1\leq i\leq n$, where $\rho'_i\neq\stopA$ and $\vec{\delta}_i = \vec{\gamma}_i = \emptystack$ by Lemma~\ref{lem:stack-len}. This implies 
\[\np{\mbox{\small$\sum$}_{i\in I\setminus\Set{k_1}}\alpha_i.\rho_i}\rho_{k_1} \pp \np{\mbox{\small$\sum$}_{j\in J\setminus\Set{k_1}}\Dual{\alpha}_j.\sigma_j}\sigma_{k_1} \ored{*}
 \np{\mbox{\small$\sum$}_{i\in I\setminus\Set{k_1}}\alpha_i.\rho_i}{\rho'_1} \pp \np{\mbox{\small$\sum$}_{j\in J\setminus\Set{k_1}}\Dual{\alpha}_j.\sigma_j}{\sigma'_1}\]
by Lemma~\ref{lem:hypSoundness}.
Let $I'=I\setminus J$ and $J'=J\setminus I$. We can reduce $\np\emptystack\rho \pp \np\emptystack\sigma$ only as follows:
\[\begin{array}{llll}
\np\emptystack\rho \pp \np\emptystack\sigma & \ored{} &
\np{\sum_{i\in I\setminus\Set{k_1}}\alpha_i.\rho_i}\rho_{k_1} \pp \np{\sum_{j\in J\setminus\Set{k_1}}\Dual{\alpha}_j.\sigma_j}\sigma_{k_1} & \mbox{by $(\CommRule)$}\\
&\ored{*}&
 \np{\sum_{i\in I\setminus\Set{k_1}}\alpha_i.\rho_i}{\rho'_1} \pp \np{\sum_{j\in J\setminus\Set{k_1}}\Dual{\alpha}_j.\sigma_j}{\sigma'_1} \\
 &\ored{}&
 \np\emptystack{\sum_{i\in I\setminus\Set{k_1}}\alpha_i.\rho_i} \pp \np\emptystack{\sum_{j\in J\setminus\Set{k_1}}\Dual{\alpha}_j.\sigma_j}
 	& \mbox{by $(\RbkRule)$}\\
	&~~\vdots&~\qquad\qquad\qquad\qquad\quad\vdots\\
	&\ored{*}&
 \np{\sum_{i\in I'}\alpha_i.\rho_i}{\rho'_n} \pp \np{\sum_{j\in J'}\Dual{\alpha}_j.\sigma_j}{\sigma'_n} \\
 &\ored{}&
 \np\emptystack{\sum_{i\in I'}\alpha_i.\rho_i} \pp \np\emptystack{\sum_{j\in J'}\Dual{\alpha}_j.\sigma_j}
 	& \mbox{by $(\RbkRule)$}
\end{array}\]
and $ \np\emptystack{\sum_{i\in I'}\alpha_i.\rho_i} \pp \np\emptystack{\sum_{j\in J'}\Dual{\alpha}_j.\sigma_j}$ is stuck since $I'\cap J'=\emptyset$.

\bigskip

Suppose $\rho = \bigoplus_{i\in I}\Dual{a}_i.\rho_i$ and $\sigma = \sum_{j\in J}a_j.\sigma_j$. If $I\not\subseteq J$ let
$k\in I\setminus J$; then we get
\[\begin{array}{llll}
\np\emptystack\rho \pp \np\emptystack\sigma & \ored{} &
\np{\emptystack}{\Dual{a}_k.\rho_k}\pp \np\emptystack{\sigma} & \mbox{by $(\tau)$} \\
& \not\!\!\ored{} &
\end{array}\]
Otherwise $I\subseteq J$ and $ {\rho_k} \ncomplyR {\sigma_k}$ for some $k\in I$. By reasoning as above we have
\[\np\emptystack\rho_k \pp \np\emptystack\sigma_k \ored{*} \np\emptystack{\rho'} \pp \np\emptystack{\sigma'} \not\!\!\ored{}\]
and 
\[  \np{\circ}{\rho_k}\pp \np{\mbox{\small$\sum$}_{j\in J\setminus\Set{k}}a_j.\sigma_j}{\sigma_k} 
 \ored{*} \np{\circ}{\rho'}\pp \np{\mbox{\small$\sum$}_{j\in J\setminus\Set{k}}a_j.\sigma_j}{\sigma'}\]
which imply
\[\begin{array}{llll}
\np\emptystack\rho \pp \np\emptystack\sigma 
& \ored{} &
\np{\emptystack}{\Dual{a}_k.\rho_k}\pp \np\emptystack{\sigma} & \mbox{by $(\tau)$} \\
&\ored{} & \np{\circ}{\rho_k}\pp \np{\sum_{j\in J\setminus\Set{k}}a_j.\sigma_j}{\sigma_k} & \mbox{by $(\CommRule)$} \\
& \ored{*} & \np{\circ}{\rho'}\pp \np{\sum_{j\in J\setminus\Set{k}}a_j.\sigma_j}{\sigma'} &\\
& \ored{} & \np\emptystack\circ \pp \np\emptystack{\sum_{j\in J\setminus\Set{k}}a_j.\sigma_j} & \mbox{by $(\RbkRule)$} \\
& \not\!\!\ored{} &
\end{array}\]
In both cases we conclude that ${\rho} \ncomplyR {\sigma}$.

\bigskip

The proof for the case  $\rho = \sum_{i\in I}a_i.\rho_i$, $\sigma = \bigoplus_{j\in J}\Dual{a}_j.\sigma_j$ is similar.
\end{proof}

\begin{theorem}[Completeness] \label{thr:completeness} If
$\rho\complyR\sigma $, then $\der \rho\complyF\sigma.$
\end{theorem}

\begin{proof}
By Theorem \ref{the:a} each computation of  
\Prove$(\der \rho\complyF\sigma) $ always terminates.
By Lemma \ref{lem:coinductiveChar} and Fact \ref{fact:provecorr}, $\rho\complyR\sigma$ implies that 
\Prove$(\der \rho\complyF\sigma) \neq$ \FAIL, and hence $\der \rho\complyF\sigma$. 
\end{proof}

\section{Related work and conclusions}


Since the pioneering work by Danos and Krivine~\cite{DK04}, reversible
concurrent computations have been widely studied. A main point is that
understanding which actions can be reversed is not trivial in a
concurrent setting, since there is no unique ``last''
action. Since~\cite{DK04}, the most common notion of reversibility in
concurrency is \emph{causal-consistent} reversibility: any action can
be undone if no other action depending on it has been executed (and
not yet undone). The name highlights the 
relation with causality, which 
makes the approach applicable even in settings where there is no
unique notion of time, but makes it quite complex.

The first calculus for which a causal-consistent reversible extension has been defined is CCS in~\cite{DK04}, using a stack of
memories for each thread. Later, causal-consistent reversible extensions have been defined by Phillips and
Ulidowski~\cite{PU07} for calculi definable by SOS rules in a general
format (without mobility), using keys to bind synchronised actions
together, and by Lanese et al.~\cite{LMS10} for the higher-order
$\pi$-calculus, using explicit memory processes to store history
information and tags to track causality. A survey of causal-consistent
reversibility can be found in~\cite{LMT14}.

In~\cite{LMSS11}, Lanese et al.\ enrich the calculus of~\cite{LMS10}
with a fine-grained rollback primitive, showing the subtleties of
defining a rollback operator in a concurrent setting.  The first papers
exploring reversibility in a context of sessions (see, e.g., \cite{LaneveP08} for a comparison between
session types and contracts) are~\cite{TY14,TY15}, by
Tiezzi and Yoshida. These papers define the semantics for reversible
sessions by adapting the approach in~\cite{LMS10}, but do not
consider compliance. Compliance has been first studied
in~\cite{BarbaneraDdL14}. We already discussed the differences between
the present work and~\cite{BarbaneraDdL14} in the Introduction.

A main point of our approach is that it exploits the fact that
contracts describe sequential interactions (in a concurrent setting)
to avoid the complexity of causal-consistent reversibility, allowing
for a simpler semantics (compared, e.g., to the one of~\cite{LMSS11}).

Similarly to our approach, long running transactions with
compensations, and in particular interacting
transactions~\cite{VriesKH10}, allow to undo past agreements. In
interacting transactions, however, a new possibility is tried when an
exception is raised, not when an agreement cannot be found as in our
case. Also, the possible options are sorted: first the normal
execution, then the compensation.  Finally, compliance of interacting
transactions has never been studied. In the field of sessions, the
most related works are probably the ones studying exceptions in binary
sessions~\cite{CarboneHY08} and in multi-party
sessions~\cite{CapecchiGY10}. As for transactions, they aim at dealing
with exceptions more than at avoiding to get stuck since an agreement
cannot be found.

We plan to investigate whether our approach can be extended to
multi-party sessions~\cite{HYC08}, the rationale being that parallelism
is controlled by the global type, hence possibly part of the
complexity due to concurrency can be avoided. The sub-behaviour
relation induced by our notion of compliance is also worth being
thoroughly studied.

\paragraph{Acknowledgments.} 
We are grateful to the anonymous reviewers for their useful remarks.
%
\label{sect:bib}
\bibliographystyle{eptcs}
\bibliography{session}



\end{document}